\documentclass[10.5pt,onecolumn,a4paper]{article}
\usepackage[latin9]{inputenc}
\usepackage{amssymb}
\usepackage{float}
\usepackage{amsthm}
\usepackage{amsmath}
\usepackage{amssymb}
\usepackage{graphicx}
\usepackage{caption}
\usepackage{subcaption}
\makeatletter
  \theoremstyle{remark}
  \newtheorem{rem}{\protect\remarkname}
  \theoremstyle{plain}
  \newtheorem{prop}{\protect\propositionname}
  \theoremstyle{plain}
  \newtheorem{lem}{\protect\lemmaname}

\usepackage{amsfonts}
\usepackage{amsmath,amsfonts,amssymb,amsthm}
\usepackage{color}
\usepackage{epsfig}

\everymath{\displaystyle}

\definecolor{bluecolor}{rgb}{0,0,1}

\usepackage{bbm}

\makeatother

  \providecommand{\lemmaname}{Lemma}
  \providecommand{\propositionname}{Proposition}
  \providecommand{\remarkname}{Remark}
  
 \headsep
 \baselineskip
\begin{document}

\title{On control of discrete-time state-dependent jump linear systems with
probabilistic constraints: A receding horizon approach\footnote{\textit{preprint submitted to Systems $\&$ Control Letters}}}

\author{Shaikshavali Chitraganti\thanks{CRAN--CNRS UMR 7039, Universit\'{e} de Lorraine, 54500 Vandoeuvre-l\'{e}s-Nancy, France.} \and Samir Aberkane\footnotemark[2] \and Christophe Aubrun\footnotemark[2] \and Guillermo Valencia-Palomo\thanks{Instituto Tecnologico de Hermosillo, Av. Tecnologico y Periferico Poniente s/n, 83170 Mexico.} \and Vasile Dragan \thanks{Instistute of Mathematics ``Simon Stoilow", Romanian Academy, RO-014700 Bucharest, Romania.}}
\date{}
\maketitle
\begin{abstract}
In this article, we consider a receding horizon control of discrete-time state-dependent jump linear systems, particular 
kind of stochastic switching systems, subject to possibly unbounded random disturbances and probabilistic state constraints. Due to a nature of the dynamical system and the constraints, we consider a one-step receding horizon. Using inverse cumulative distribution function, we convert the probabilistic state constraints to deterministic constraints, and obtain a tractable deterministic receding horizon control problem. We consider the receding control law to have a linear state-feedback and an admissible offset term. We ensure mean square boundedness of the state variable via solving linear matrix inequalities off-line, and solve the receding horizon control problem on-line with control offset terms. We illustrate the overall approach applied on a macroeconomic system.
\end{abstract}

\section{Introduction}
Dynamical systems subject to random abrupt changes, such as manufacturing
systems, networked control systems, economics and finance \textit{etc}.,
can be modelled adequately by random jump linear systems (RJLSs). RJLSs
are a particular kind of stochastic switching systems, which consists
of a set of linear systems, also called multiple modes, and the switching
among them is governed by a random jump process.

A notable class of RJLSs is Markov jump linear system (MJLS) in which
the underlying random jump process is a finite state Markov chain (or a
finite state Markov process). Many important results related to stability,
control, and applications of such systems have been investigated in
the literature, for instance in \cite{boukasbook2006}, \cite{book_do2005discrete},
\cite{fang}, \cite{feng}, \cite{huang2013analysis}, \cite{li2013stochastic} etc. Almost all the works
related to MJLSs assume that the underlying random jump process is time-homogeneous/time-inhomogeneous Markov,
which is a restrictive assumption. 

In this article, we deal with a class of RJLSs in which the evolutions
of the random jump process depends on the state variable, and are
referred to as \textit{state-dependent jump linear systems} (SDJLS).
In the following, we list some motivations for SDJLS modelling of dynamical
systems. In the analysis of random breakdown of components, the age,
wear, and accumulated stress of a component affect its failure rate,
for instance. Thus, it can be assumed that the failure rate of a component
is dependent on state of the component at age $t$ \cite{bergman},
where the state variable may be an amount of wear, stress etc. As
an another instance, in \cite{motivation1}, a state-dependent Markov
process was utilized to describe the random break down of cylinder
lines in a heavy-duty marine diesel engine. Also, we can examine a
stock market with situations: up and down, and the transitions between
the situations can be dependent on state of the market, where the
state variable may be general mood of investors and current economy
etc. Also, a state-dependent regime switching model was considered
in \cite{motivation2} to model financial time series. One can find other instances or examples 
of SDJLS modelling in the literature. 

The studies of stability and control of SDJLSs have been scanty in the literature. A study of hybrid switching
diffusion processes, a kind of continuous-time state-dependent jump non-linear systems with diffusion,
was considered in \cite{yin2009hybrid} by treating existence, uniqueness, stability of the solutions etc.
For RJLSs, a state and control dependent random jump process was considered
in \cite{sworder1973}, where the authors used stochastic maximum
principle to obtain optimal control for a given objective function.
SDJLS modelling of flexible manufacturing system was proposed in \cite{boukas1990},
and dynamic programming is used to obtain an optimal input which minimizes
the mentioned cost. A state-dependent jump diffusion modelling of a
production plant was considered in \cite{filar2001} to obtain an optimal
control. In the sequel, we bring back the attention to the main ingredients
of the problem that we address in this article.

In this article, we consider that the SDJLS is affected by possibly unbounded stochastic
disturbances, and the perfect state information is assumed. We also deal with constraints on different variables of the system  that is inherent in all practical systems. Model predictive control (MPC), also called receding horizon control (RHC),  is an effective control algorithm that has a great potential to handle input and/or state constraints, for problems across different
disciplines. MPC is a form of control in which the current input is
obtained at each sampling time by solving on-line a finite horizon
optimal control problem in the presence of constraints, using the
current information and the predicted information over the finite
horizon. Normally more than one input is obtained at the current sampling
time, however, only the first controller input will be implemented
to the plant. At the next sampling time, these actions will be repeated,
that is why the MPC is also called the RHC. One can refer to \cite{LS_MPC_camacho2004model},
\cite{book_maciejowski2002predictive}, \cite{LS_MPC_mayne2000constrained},
\cite{LS_MPC_rawlings1993stability} etc., for classic contributions
in the MPC literature. In the context of RJLSs, of late, the RHC scheme has been extended to discrete-time MJLSs. For discrete-time MJLSs, the RHC with hard symmetric constraints and bounded uncertainties in system parameters was dealt
by  \cite{MPC_MJLS_lu2013}, \cite{MPC_MJLS_wen2013}, where the constraints and the objective function were posed as sufficient conditions in terms of linear matrix inequalities (LMIs) to be solved at each sampling time; for the similar case without constraints,  the RHC was addressed by \cite{LSC_park2002robust} following the similar approach. Also, for unconstrained state and control
input, optimality of the RHC was addressed via dynamic programming \cite{vargas2005optimality}, variational
methods \cite{vargas2004receding}, solving Riccati equations \cite{do1999receding},
etc.

One major issue in the presence of unbounded
disturbances is that the RHC cannot necessarily guarantee the satisfaction
of constraints. For instance, an additive unbounded disturbance eventually drives the state
variable outside any bounded limits, no matter how arbitrarily large they may be.  A possible alternative is to consider the satisfaction of constraints in stochastic manner, which allow occasional
constraint violations. In this direction, recent approaches \cite{LS_MPC_primbs2009stochastic},
\cite{Cannon2009}, \cite{LS_SOFTMPC_korda2011strongly}, and the references therein treat the RHC of discrete-time linear systems with stochastic constraints. For discrete-time MJLSs, in case of perfect state availability, a linear quadratic regulator problem with second moment constraints was considered in \cite{vargas2013second}, where the entire problem was converted to a set of LMIs. However, to the best of the authors' knowledge, the RHC of discrete-time SDJLSs with probabilistic constraints has not been examined.

In this article, we address a one-step RHC of discrete-time SDJLSs with additive Gaussian process noise (its distribution has an unbounded support), and probabilistic state constraints under perfect state availability. We would like to highlight several challenges in our problem set-up. First, in the presence of additive process noise with unbounded support, it impossible to guarantee hard bounds on the state, and also on the linear state-feedback control. Second, one needs to pre-stabilize the system before addressing the RHC problem. Third,
one needs to obtain a tractable representation of the RHC problem in the presence of probabilistic constraints.

Our approach along with main contributions in this article can be listed as follows. In our
problem set-up, we consider the control to have a linear state-feedback and an offset term
\cite{bemporad1999robust}, \cite{mayne2005robust}, where linear state-feedback gains are computed off-line for pre-stabilization and admissible offset terms are computed on-line to solve the RHC problem. In the presence of unbounded
process noise, it is not possible to ensure hard bounds on the state and the control variables that
follow state-feedback law, thus we consider the state variable to be probabilistically constrained
and the control to be unconstrained, except for the offset term. Using inverse cumulative distribution function, we convert the probabilistic state constraints to deterministic constraints and the overall RHC problem is replaced by a tractable deterministic RHC problem.

To summarize, for SDJLS subject to possibly unbounded random disturbances and probabilistic state constraints,
our contributions in this article are: 
\begin{itemize}
\item pre-stabilizing the system state by a state-feedback controller in means square sense,

\item implementing a one-step RHC scheme on-line with probabilistic constraints
on the state variable, which are converted to deterministic constraints.

\end{itemize}
For illustration, we apply our approach to a macroeconomic situation. 

The article is organized as follows. Section \ref{sec:Mathematical-Model}
presents the problem setup. We present the pre-stabilization of the system by a state
feed-back controller in section \ref{sec:Pre-stabilization}.
We convert the probabilistic constraints to suitable deterministic
constraints in section \ref{sec:prob_constraints}. We give a one-step RHC scheme
with probabilistic constraints in section \ref{sec:One-step-RHC-with}.
Section \ref{sec:Examples} presents an illustrative example followed by conclusions in section \ref{sec:conclusions}. Finally, we
give majority of the proofs in the Appendix to improve readability.

\textit{Notation: }Let $\mathbb{R}^{n}$ denotes the $n$-dimensional
real Euclidean space and  $\mathbb{N}_{\ge 0}$  the set of non-negative integers.  For a matrix $A$, $A^{T}$ denotes the transpose, $\lambda_{\text{min}}(A)$ $\left(\lambda_{\text{max}}(A)\right)$
the minimum (maximum) eigenvalue and $\mathrm{tr}(A)$
the trace of $A$. The standard vector norm in
$\mathbb{R}^{n}$ is denoted by $\Vert.\Vert$ the corresponding
induced norm of a matrix $A$ by $\Vert A\Vert$. Given a matrix $L$,
$L\succ0$ (or $L\prec0$) denotes that the matrix $L$ is positive
definite (or negative definite). Given two matrices $L$ and $M$,
$L\ge M$ (or $L\le M$) denotes the element wise inequalities. Symmetric
terms in block matrices are denoted by $\ast.$ A matrix product $AP\star$
denotes $APA^{T}.$ The identity matrix of dimension $n\times n$ is denoted by $\mathbb{I}_{n}$. The diagonal matrix formed from its vector arguments is denoted by diag$\{.\}$. The underlying probability space is denoted by  $(\Omega,\mathcal{F},\mathrm{Pr})$ where $\Omega$ is the space of elementary events, $\mathcal{F}$ is a $\sigma$-algebra, and $\mathrm{Pr}$ is the probability measure. The mathematical expectation of a random variable $X$ is denoted by $\mathbb{E}[X]$.  Further notation will be introduced when required.

\section{Problem setup\label{sec:Mathematical-Model}}

Consider a discrete-time SDJLS:
\begin{equation}
x_{k+1}=A_{\theta_{k}}x_{k}+B_{\theta_{k}}u_{k}+E_{\theta_{k}}w_{k},\quad k\in \mathbb{N}_{\ge 0}, \label{sys:djls_noise}
\end{equation}
with the state $x_{k}\in\mathbb{R}^{n}$, the control input $u_{k}\in\mathbb{R}^{m}$, the system mode $\theta_k$.
Let $A_{\theta_{k}},$ $B_{\theta_{k}}$ be the system matrices of appropriate
dimensions that are assumed to be known, and without loss of generality we assume that $E_{\theta_{k}}$ are
identity matrices. Perfect state information and the mode availability is assumed at each time $k\in \mathbb{N}_{\ge 0}$. We set $\mathcal{F}_{k}$ the $\sigma$-algebra generated by $\{(\theta_t,x_t); 0\leq t\leq k\}$. Regarding the different stochastic processes,
the following assumptions are made: \\

\noindent \textbf{Assumption 1.}  Let $w_{k}\in \mathbb{R}^n$ be an $n-$dimensional standard Gaussian distributed random vector i.e.  $w_{k}\sim\mathcal{N}(0,\mathbb{I}_{n})$, independent of $\mathcal{F}_k$.\\

\noindent \textbf{Assumption 2.} Let the mode $\{\theta_{k}\}_{k\in \mathbb{N}_{\ge 0}}\in S=\{1,2,\cdots,N\}$ be a finite-state
random jump process described by
\begin{align}
\mathrm{Pr}\{\theta_{k+1}=j|\mathcal{F}_{k}\}=\mathrm{Pr}\{\theta_{k+1}=j|(\theta_{k},x_{k})\}=\pi_{\theta_{k}j}(x_{k})=\begin{cases}
\lambda_{\theta_{k}j}, & \mathrm{if}\, x_{k}\in\mathcal{C}_{1},\\
\mu_{\theta_{k}j}, & \mathrm{if}\, x_{k}\in\mathcal{C}_{2},\,
\end{cases}\label{eq:p_tran}
\end{align}
where $\lambda_{ij}\ge0\,(\mu_{ij}\ge0)$ is the transition probability
from a state $i$ to a state $j$, and ${\textstyle \sum\nolimits _{j=1}^{N}}\lambda_{ij}=1\left({\textstyle \sum\nolimits _{j=1}^{N}}\mu_{ij}=1\right)$.
Here we further assume that $\mathcal{C}_{1}\cup\mathcal{C}_{2}=\mathbb{R}^{n}$,
$\mathcal{C}_{1}\cap\mathcal{C}_{2}=\phi$.\\

If we observe the SDJLS (\ref{sys:djls_noise}), the actual state consists of two parts: the state $x_k$ and the mode $\theta_k$. If the mode $\theta_k$ is fixed for all $k\in \mathbb{N}_{\ge 0}$, then the SDJLS (\ref{sys:djls_noise}) is equivalent to the standard linear time-invariant system. Thus the role of the mode $\theta_k$ is important in the SDJLSs, and in general in any switching systems. With this background, we observe the mode dependence through out the article: in designing state-feedback gains, in handling probabilistic constraints and in weighting matrices in the objective function.

The sets  $\mathcal{C}_{1}$ and $\mathcal{C}_{2}$ in (\ref{eq:p_tran}) are described explicitly as follows. We consider the set $\mathcal{C}_{1}$ as a polyhedral intersection
of finite constraints on the state variable and $\mathcal{C}_{2}$
as $\mathbb{R}^{n}-\{\mathcal{C}_{1}\}$. So, in this article, the
states of the system that belong to $\mathcal{C}_{1}$ are specifically
denoted by
\begin{align}
x_{k}\in\mathcal{C}_{1}\Longleftrightarrow Gx_{k}\le H,\label{eq:polyhedral}
\end{align}
where $G\in\mathbb{R}^{r\times n}$ and $H\in\mathbb{R}^{r}$. We can regard $\mathcal{C}_1$ as a desirable set, and ideally we want the system state to be constrained in $\mathcal{C}_1$. However, the state variable being affected by Gaussian noise that is unbounded in distribution, it is impossible to constrain the state to the set $\mathcal{C}_1$ at all times. Thus, we consider probabilistic state constraints, and in section \ref{sec:prob_constraints}, we describe in detail the type of probabilistic state constraints that we deal in this article.
 
In our RHC problem, we address the minimization of the following one-step objective function
\begin{align}
&J_{k}=\mathbb{E} \Big[x_{k}^{T}Q_{\theta_{k}}x_{k}+u_{k}^{T}R_{\theta_{k}}u_{k}+x_{k+1}^{T}\Psi_{\theta_{k+1}}x_{k+1}|\mathcal{F}_{k}\Big] \label{obj:func}
\end{align}
subject to probabilistic state constraints and the control input that follows an affine state-feedback law. Here $Q_{i}\succ 0$, $R_{i}\succ 0,$ $\Psi_{i}\succ 0$ are assumed to be known for $i\in S.$ In this formulation, the terminal
cost depend on the matrices $\Psi_{(.)}$ and the state of the system at time $k+1$. The choice of one-step prediction in (\ref{obj:func}) is explained in detail in section \ref{sec:One-step-RHC-with} in perspective of the system and the constraints.
\subsection{Control policies}
We parametrize the control input as (inspired by \cite{bemporad1999robust}, \cite{mayne2005robust}) 
\begin{equation}
u_{k}=K_{\theta_{k}}x_{k}+\nu_{k},\label{eq:stabilization_input}
\end{equation}
where feed-back gains $K_{\theta_{k},}$ for $\theta_{k}\in S$, will
be designed to pre-stabilize the system in mean square sense that will be given the next section. An on-line minimization of a $J_k$ in (\ref{obj:func}) subject to constraints will be
carried out by an appropriate $\nu_{k}$.\\

\begin{rem}
Notice that the state-feedback control input (\ref{eq:stabilization_input})
is not bounded a priori. The reason is as follows. From (\ref{sys:djls_noise}), the state $x_k$ in (\ref{eq:stabilization_input})
is a function of Gaussian noise. Since Gaussian distribution
has an unbounded support, we cannot assume the a priori boundedness
of $u_{k}$ (\ref{eq:stabilization_input}). Thus, we just consider $\nu_k$ to be bounded to obtain a tractable solution of minimization of $J_k$ in (\ref{obj:func}) subject to constraints.
\end{rem}
Thus, we give an explicit assumption on $\nu_k$ in the following.\\

\noindent \textbf{Assumption 3.} \label{assume:input_compact} We assume that
$\nu_{k}$ belong to a compact set $\mathbb{U}\subset\mathbb{R}^{m}$.\\ 
\begin{rem}
\label{remark:input_measurable}Note that $\nu_{k}$ is obtained by solving $\mathcal{P}_3$ on-line that will be introduced in section \ref{sec:One-step-RHC-with}. It will be observed that the computation of $\nu_{k}$ relies on the available information $(\theta_k,x_k)$, for each $k\in \mathbb{N}_{\ge 0}$. Thus $\nu_k$ can be regarded as $\mathcal{F}_{k}$-measurable random variable taking values in $\mathbb{U}$.\\
\end{rem}
Substituting (\ref{eq:stabilization_input}) into (\ref{sys:djls_noise}),
we obtain the closed loop system as
\begin{equation}
x_{k+1}=\tilde{A}_{\theta_{k}}x_{k}+B_{\theta_{k}}\nu_{k}+w_{k},\label{sys:stabilized_input}
\end{equation}
 with $\tilde{A}_{\theta_{k}}=A_{\theta_{k}}+B_{\theta_{k}}K_{\theta_{k}}.$
 
\begin{rem}
In this article, we consider the linear state-feedback control input (\ref{eq:stabilization_input}) for pre-stabilization, where we could not ensure hard bounds on the control because of the stochastic unbounded disturbance.  To incorporate bounded controls in the class of causal feedback policies, one alternative could be employing an affine disturbance feedback control with non-linear saturation of noise terms \cite{hokayem2010stable}, \cite{hokayem2012stochastic}. In this approach, the disturbance implicitly depends on the output, hence on the state in a causal manner. However, in this case the system under consideration needs to be internally stable.
\end{rem}

In the next section, we synthesize the state-feedback gains $K_{\theta_{k}}$ in (\ref{eq:stabilization_input}), which pre-stabilize
(\ref{sys:stabilized_input}) in bounding the state-trajectories in mean square sense.

\section{Pre-stabilization results\label{sec:Pre-stabilization}}

In this section, we pre-stabilize  the system (\ref{sys:stabilized_input}) with state-feedback gains $K_{\theta_{k}}$ by solving sufficient LMIs off-line.

In general, in the presence of
unbounded disturbances, it is difficult to guarantee or establish
stability. It is impossible to expect asymptotic stability to origin.
However, we synthesize a state-feedback controller
that bound the second moment of the state-trajectories. The results of this section rely on a stochastic version of Lyapunov's second method \cite{book_do2005discrete}, \cite{book_dragan2010mathematical} by the choice of mode-dependent Lyapunov function $V(x_k,\theta_k)=x_k^T P_{\theta_k}x_k$ with $P_{\theta_k}\succ 0, \forall \theta_k \in S$. In \cite{book_dragan2010mathematical}, the authors studied the stochastic stability of discrete-time MJLSs subject to uncertainties in system parameters, and we extend these results to examine the stochastic stability of SDJLSs (\ref{sys:stabilized_input}),  where the underlying random jump process is state-dependent, which is the theme of our approach. We begin with the following definitions.\\

\noindent Let $\mathcal{T}\subset\mathbb{R}^{m\times n}$ be a linear
subspace of real matrices. Set $\mathcal{T}^{N}=\mathcal{T}\times\mathcal{T}\times\cdots\times\mathcal{T}$.
Also, let $\mathcal{S}_{n}\subset\mathbb{R}^{n\times n}$ be a linear
subspace of real symmetric matrices. Set $\mathcal{S}_{n}^{N}=\mathcal{S}_{n}\times\mathcal{S}_{n}\times\cdots\times\mathcal{S}_{n}$.\\ 


We give a proposition to pre-stabilize the system (\ref{sys:stabilized_input}) in the following.

\begin{prop} \label{theorem:msqs}
\noindent If there exist matrices $K=(\begin{array}{ccc}
K_{1} & \cdots & K_{N}\end{array})\in\mathcal{T}^{N}$, $P=\left(\begin{array}{ccc}
P_{1} & \cdots & P_{N}\end{array}\right)\in\mathcal{S}_{n}^{N}$, with $P_{i}\succ0$, $1\leq i\leq N$ such that 
\begin{equation}
\tilde{A}_{i}^{T}\left(\sum\nolimits _{j=1}^{N}\lambda_{ij}P_{j}\right)\tilde{A}_{i}-P_{i}=\mathcal{E}_{i}^{1}\prec0,\label{LMI1}
\end{equation}
and 
\begin{equation}
\tilde{A}_{i}^{T}\left(\sum\nolimits _{j=1}^{N}\mu_{ij}P_{j}\right)\tilde{A}_{i}-P_{i}=\mathcal{E}_{i}^{2}\prec0.\label{LMI2}
\end{equation}
then the trajectories of (\ref{sys:stabilized_input}) satisfy
the estimates:
\end{prop}
\begin{equation}\label{eq:msqs}
\mathbb{E}\left[\Vert x_{k}\Vert^{2}\right]\leq\beta_{1}\Vert x_{0}\Vert^{2}q^{k}+\beta_{2},\;\forall\; k\in \mathbb{N}_{\ge 0},
\end{equation}
\textit{for some $\beta_{1}\geq1$, $\beta_{2}>0$ and $q\in(0,1)$.}
\begin{proof}
Given in Appendix A.
\end{proof}
In the sequel, we provide a proposition to synthesize gains $K_{\theta_{k}},$
for $\theta_{k}\in S.$ 
\begin{prop}
\label{theorem:stabilization_SS} If there exist matrices $X_{i}\succ0$
and $Y_{i}$, $1\leq i\leq N$, such that 
\end{prop}
\begin{equation}
\begin{bmatrix}-X_{D} & \Lambda_{i}\left(A_{i}X_{i}+B_{i}Y_{i}\right)\\
\ast & -X_{i}
\end{bmatrix}\prec0,\label{lmi:stabiization1}
\end{equation}
\textit{and} 
\begin{equation}
\begin{bmatrix}-X_{D} & \Gamma_{i}\left(A_{i}X_{i}+B_{i}Y_{i}\right)\\
\ast & -X_{i}
\end{bmatrix}\prec0,\label{lmi:stabiization2}
\end{equation}
\textit{where} 
\begin{equation*}
\begin{cases}
\Lambda_{i}  =\begin{bmatrix}\sqrt{\lambda_{i1}}\mathbb{I}_{n} & \cdots & \sqrt{\lambda_{iN}}\mathbb{I}_{n}\end{bmatrix}^{T},\\
\Gamma_{i}  =\begin{bmatrix}\sqrt{\mu_{i1}}\mathbb{I}_{n} & \cdots & \sqrt{\mu_{iN}}\mathbb{I}_{n}\end{bmatrix}^{T},\\
X_{D} =\text{diag}\{X_{1},\cdots,X_{N}\},
\end{cases}
\end{equation*}
\textit{then the trajectories of the system} (\ref{sys:stabilized_input})
\textit{satisfy} (\ref{eq:msqs}), \textit{and $K_{i}$ is given by}
\begin{equation}
K_{i}=Y_{i}X_{i}^{-1}.\label{controller:stateFB}
\end{equation}

\begin{proof}
Briefly, we explain the steps to arrive at the result. Consider the
Schur complement of (\ref{LMI1}), (\ref{LMI2}), and let $X_{i}\triangleq P_{i}^{-1}$.
By the congruent transformation of Schur complement of (\ref{LMI1}),
(\ref{LMI2}) by diag$\{X_{1},\cdots,X_{N},X_{i}\}$, and applying
a change of variable $Y_{i}\triangleq K_{i}X_{i}$ one obtains (\ref{lmi:stabiization1}),
(\ref{lmi:stabiization2}). Hence, if (\ref{lmi:stabiization1}),
(\ref{lmi:stabiization2}) are feasible and using proposition \ref{theorem:msqs}, it follows that the trajectories of (\ref{sys:stabilized_input}) satisfy the estimates (\ref{eq:msqs}). Hence the proof is complete.
\end{proof}
\begin{rem}
Though the conditions to obtain the state-feedback gains $K_{\theta_k}$, for $\theta_k\in S$, in proposition \ref{theorem:stabilization_SS} seems computationally heavy, we perform these computations off-line, and only once. Once we obtain the state-feedback gains $K_{\theta_k}$, we use them as a lookup table, where depending on the mode $\theta_k$ at each time $k\in\mathbb{N}_{\ge 0}$, the corresponding value of the state-feedback gain is used in updating the state in (\ref{sys:stabilized_input}).
\end{rem}
\section{Probabilistic constraints\label{sec:prob_constraints}}
In this section, we consider the probabilistic state constraints that we address in this article, and replace them by deterministic constraints.

We mentioned in section \ref{sec:Mathematical-Model}, that it is impossible to constrain the state variable to the set $\mathcal{C}_{1},$ which is polyhedron $Gx_{k}\le H$, where $G\in\mathbb{R}^{r\times n}$ and $H\in\mathbb{R}^{r}$. Rather,
let us consider probabilistic constraint

\begin{equation}
\text{Pr}\{Gx_{k}\le H\}\ge\xi,\label{constraint:multivariate}
\end{equation}
which portrays the probabilistic scenario of staying in the polyhedron.
Here $\xi$ represents the level of constraint satisfaction, which needs
to be high enough ($\xi>>0$). However, without loss of generality, we assume that $\xi$ may take any
values in the interval $[0,1]$.  In the sequel, we try to obtain deterministic
constraints that imply the probabilistic constraint (\ref{constraint:multivariate}).
We say it as converting probabilistic constraint to deterministic
ones. A sufficient condition to satisfy (\ref{constraint:multivariate})
is given by \cite{LS_SOFTMPC_korda2011strongly},

\begin{equation}
\text{Pr}\{Gx_{k+1}\le H|\mathcal{F}_k\}\ge\xi. \label{constraint:multivariate_final}
\end{equation}
Now, we present a sufficient deterministic condition to satisfy (\ref{constraint:multivariate_final})
via approximation of inscribed ellipsoidal constraints \cite{van2001lmi},
\cite{van2002conic}. 
\begin{lem}
\label{theorem:multivariate}Let $\delta=\sqrt{F_{n,Chi}^{-1}(\xi)}$,
where $F_{n,Chi}^{-1}(\xi)$ denotes Chi-square inverse cumulative
distribution function with a given probability $\xi$ and $n$ degrees
of freedom. If

\begin{equation}
G_{j}(\tilde{A}_{\theta_{k}}x_{k}+B_{\theta_{k}}\nu_{k})\le H_{j}-\Vert G_{j}\Vert_{2}\delta,\,\text{ for }\,1\le j\le r,\label{constraint:chi_square}
\end{equation}
 then $\,\mathrm{Pr}\{Gx_{k+1}\le H|\mathcal{F}_k\}\ge\xi$, where $G_j$ and $H_j$ are $j^{\text{th}}$ rows of $G$ and $H$ respectively.\end{lem}
\begin{proof}
The proof is given on the similar lines of \cite{van2001lmi}, \cite{van2002conic}
and outlined in Appendix B.\end{proof}

Note that (\ref{constraint:chi_square}) is an over conservative
condition, because the deterministic condition (\ref{constraint:chi_square})
implies (\ref{constraint:multivariate_final}), which finally implies
(\ref{constraint:multivariate}). Even for this special case, we could not obtain an equivalent representation of (\ref{constraint:multivariate_final}), which is non-convex in general, where it is hard to find its feasible region \cite{prekopa1995stochastic}. So, alternatively, we propose the individual probabilistic constraints of type
\begin{equation}
\text{Pr}\{G_{j}x_{k}\le H_{j}\}\ge\xi,\:1\le j\le r,\label{constraint:univariate}
\end{equation}
where $G_{j}$ and $H_{j}$ represent the $j^{\text{th}}$ row of
$G$ and $H$ respectively. With given $\xi$, the constraints (\ref{constraint:univariate})
offer satisfaction of each individual constraint of the polyhedron
probabilistically, but with more constraint violations \cite{arnold2013mixed} than (\ref{constraint:multivariate})  that will also be observed in section \ref{sec:Examples}. However, we consider the individual probabilistic constraints (\ref{constraint:univariate}), because they are simpler to handle and in general convex \cite{prekopa1995stochastic}, \cite{arnold2013mixed}. Similar to the above treatment, a sufficient
condition to satisfy (\ref{constraint:univariate}) is given by 
\begin{equation}
\text{Pr}\{G_{j}x_{k+1}\le H_{j}|\mathcal{F}_k\}\ge\xi,\quad1\le j\le r.\label{constraint:univariate_final}
\end{equation}
By (\ref{sys:stabilized_input}), one obtains,
\begin{align}
\text{Pr}\{G_{j}(\tilde{A}_{\theta_{k}}x_{k}+B_{\theta_{k}}\nu_{k}+w_{k})\le H_{j}\} & \ge\xi\label{constraint:remark_univar}\\
\Longleftrightarrow \text{Pr}\{G_{j}w_{k}\le H_{j}-G_{j}(\tilde{A}_{\theta_{k}}x_{k}+B_{\theta_{k}}\nu_{k})\} & \ge\xi\nonumber \\
\Longleftrightarrow F_{G_{j}w_{k}}(H_{j}-G_{j}(\tilde{A}_{\theta_{k}}x_{k}+B_{\theta_{k}}\nu_{k})) & \ge\xi\nonumber \\
\Longleftrightarrow G_{j}(\tilde{A}_{\theta_{k}}x_{k}+B_{\theta_{k}}\nu_{k}) & \le H_{j}-F_{G_{j}w_{k}}^{-1}(\xi),\label{constraint:det_univariate}
\end{align}
where $F_{G_{j}w_{k}}(.)$ and $F_{G_{j}w_{k}}^{-1}(.)$ are the cumulative
distribution and inverse cumulative distribution of the random variable
$G_{j}w_{k}$ respectively. 
\begin{rem}
The probabilistic constraints (\ref{constraint:univariate}) result
in handling with uni-variate Gaussian random variables $G_{j}w_{k}$
when converting to deterministic constraints (\ref{constraint:det_univariate}),
which is straightforward. In this case, given $\xi$, $F_{G_{j}w_{k}}^{-1}(\xi)$
can easily be obtained.
\end{rem}
Observe that the conditions (\ref{constraint:univariate_final}) and
(\ref{constraint:det_univariate}) are equivalent, which imply (\ref{constraint:univariate}).
\begin{rem}
The probabilistic constraints (\ref{constraint:multivariate}) and
(\ref{constraint:univariate}) are two different ways of treating
the constraints $Gx_{k}\le H$ in a probabilistic fashion. We consider the probabilistic constraints (\ref{constraint:univariate_final})
because of the simplicity and low conservatism involved.
\end{rem}

\section{A One-step RHC with probabilistic state constraints\label{sec:One-step-RHC-with}}
In this section, we provide a one-step RHC problem  with the objective function (\ref{obj:func}) subject to the
probabilistic constraints (\ref{constraint:univariate_final}) and the state-feedback control input (\ref{eq:stabilization_input}).  

At each $k\in \mathbb{N}_{\ge 0}$, we consider the following one-step RHC problem
\begin{align}
\boldsymbol{\mathcal{P}_{1}:} & \min_{\nu_{k}}J_{k}=\mathbb{E}\Big[x_{k}^{T}Q_{\theta_{k}}x_{k}+u_{k}^{T}R_{\theta_{k}}u_{k}+x_{k+1}^{T}\Psi_{\theta_{k+1}}x_{k+1}|\mathcal{F}_{k}\Big]\nonumber \\
 & \text{s.t.}\; (\ref{sys:djls_noise}), (\ref{eq:p_tran}),\nonumber\\
 & \quad\;\: \text{Pr}\big\{G_{j}x_{k+1}\le H_{j}|\mathcal{F}_k\big\}\ge\xi,\, 1\le j\le r,\,\xi\in[0,1],\label{eq:prob_constraint}\\
 & \quad\;\: u_{k}=K_{\theta_{k}}x_{k}+\nu_{k},\\
 & \quad\;\:\nu_{k}\in\mathbb{U}.\label{eq:input_con}
\end{align}
 
\begin{rem}
In $\mathcal{P}_{1}$, we choose the prediction horizon to be one because of the probabilistic constraints (\ref{eq:prob_constraint}) and the system (\ref{sys:stabilized_input}). In section \ref{sec:prob_constraints}, we obtained a deterministic equivalence of the constraints (\ref{eq:prob_constraint}) in terms of  the state variable and the mode at time $k$ that are known. In general, the larger prediction horizon result in better performance and more computational burden depending on the system \cite{book_maciejowski2002predictive}. The choice of the prediction horizon depends on the performance requirements and computational capability. Suppose, if we consider a multi-step prediction, the probabilistic constraints in $\mathcal{P}_{1}$ look like $\text{Pr}\big\{G_{j}x_{k+N}\le H_{j}|\mathcal{F}_k\big\}\ge\xi, \text{for }N\ge 2$.  By proceeding  with the similar approach of section \ref{sec:prob_constraints}, we can obtain an equality similar to (\ref{constraint:remark_univar}) that contain additional unknown random variables $\theta_{k+m},\,1\le m\le N-1$, where it is not possible to obtain its deterministic equivalence.  Thus, we choose a one-step prediction horizon to obtain a deterministic equivalence of the probabilistic constraints  (\ref{eq:prob_constraint}) for tractability of $\mathcal{P}_1$. At this point, we give a brief review of some works in the literature of RJLSs with multi-step prediction. In case of perfect state availability, without process noise, a multi-step prediction horizon was considered for discrete-time MJLSs with hard symmetric input and state constraints \cite{MPC_MJLS_lu2013}, \cite{MPC_MJLS_wen2013}, where an upper bound of the objective function  was minimized and the overall problem with the constraints was converted to sufficient LMIs. The approach was based on \cite{eco_costa1999constrained}, where the state variable was restriced to an invariant ellipsoid in the constrained space and the techniques of LMIs and Lyapunov method were utilized. However, it brings a computational burden because of the additional LMIs, and more importantly  it reduces the feasibility region drastically depending on the invariant ellipsoid. To avoid the sufficiency, the authors in \cite{vargas2013second} directly considered ellipsoidal constraints (in terms of second moment constraints) for discrete-time MJLSs with multi-step prediction, where the constraints were replaced by a set of LMIs. 
\end{rem}

It can be possible that the $\mathcal{P}_{1}$ become infeasible when the constraints
(\ref{eq:prob_constraint}) are tight with a given admissible input.
To resolve this issue, the constraints can be relaxed by an additional
slack variable $\rho_{k}\ge 0$ as \cite{schwarm1999chance},

\begin{equation}
\text{Pr}\big\{G_{j}x_{k+1}\le H_{j}+\mathbf{1}\rho_{k}|\mathcal{F}_k\big\}\ge\xi, \label{constraine:slack_old}
\end{equation}
where $\mathbf{1}$ denotes a column vector of appropriate size, which
contain all ones. Thus from (\ref{constraint:det_univariate}),

\begin{equation}
G_{j}\big(\tilde{A}_{\theta_{k}}x_{k}+B_{\theta_{k}}\nu_{k}\big)\le H_{j}-F_{G_{j}w_{k}}^{-1}(\xi)+\mathbf{1}\rho_{k}.\label{constraint:slack_final}
\end{equation}
The addition of a variable $\rho_{k}$ can be compensated by adding
a variable $\alpha\rho_{k}$ to the objective function in  $\mathcal{P}_{1}$. In
particular, $\alpha>0$ should be chosen as a very high
value, which act as a penalty to discourage the use of slack variable
\cite{feasibilitychinneck2007}. 

Thus  $\mathcal{P}_{1}$  will be converted to 
\begin{align}
\boldsymbol{\mathcal{P}_{2}:} & \min_{\nu_{k},\rho_{k}}J_{k}=x_{k}^{T}Q_{\theta_{k}}x_{k}+u_{k}^{T}R_{\theta_{k}}u_{k}+\mathbb{E}\big[x_{k+1}^{T}\Psi_{\theta_{k+1}}x_{k+1}|\mathcal{F}_{k}\big]+\alpha\rho_{k}\nonumber \\
& \text{s.t.}\; (\ref{sys:djls_noise}), (\ref{eq:p_tran}),\nonumber\\
 & \quad\;\: G_{j}\big(\tilde{A}_{\theta_{k}}x_{k}+B_{\theta_{k}}\nu_{k}\big)\le H_{j}-F_{G_{j}w_{k}}^{-1}(\xi)+\mathbf{1}\rho_{k},\, 1\le j\le r,\,\xi\in[0,1],\label{constraint:slack_RHC2}\\
 & \quad\;\: u_{k}=K_{\theta_{k}}x_{k}+\nu_{k},\\
 & \quad\;\:\nu_{k}\in\mathbb{U}, \\
 & \quad\;\:\rho_{k}\ge0.\label{eq:epsilon}
\end{align}

Notice that due to (\ref{constraint:slack_RHC2}), (\ref{eq:epsilon}), $\mathcal{P}_{2}$ is always feasible.

The main task is to minimize the objective function in $\mathcal{P}_{2}$ by
a proper choice of $\nu_{k}.$ To accomplish this task, we present
$\mathcal{P}_{2}$ in a tractable fashion in the following.

\subsection{Formulating $\mathcal{P}_{2}$ as a tractable problem}

In order to write the $\mathcal{P}_{2}$ in terms of tractable terms, consider the
term $\mathbb{E}\left[x_{k+1}^{T}\Psi_{\theta_{k+1}}x_{k+1}|\mathcal{F}_{k}\right]$.
From (\ref{sys:stabilized_input}), $x_{k+1}$ is a random vector with

\[
\mathbb{E}\left[x_{k+1}|\mathcal{F}_{k}\right]=\tilde{A}_{\theta_{k}}x_{k}+B_{\theta_{k}}\nu_{k},
\]
and
\[
\begin{aligned}\mathbb{E}\left[\left\{x_{k+1}-\left(\tilde{A}_{\theta_{k}}x_{k}+B_{\theta_{k}}\nu_{k}\right)\right\}1\star|\mathcal{F}_{k}\right]=\mathbb{E}\left[w_{k}w_{k}^{T}\right]=\mathbb{I}_{n}.\end{aligned}
\]
And now, let us consider the following lemma which is a classic result,
\begin{lem}
\label{lem:lemma1} Let $x$ be a random vector with mean $m$ and
covariance matrix $R$, and let $K$ be a given matrix of suitable
dimensions. Then
\begin{equation}
\mathbb{E}\left[x^{T}Kx\right]=m^{T}Km+\mathrm{tr}(KR).\label{eq:lemm1}
\end{equation}
\end{lem}

The result (\ref{eq:lemm1}) and the law of total probability gives
that
\begin{equation}
\begin{aligned}\mathbb{E}\Big[x_{k+1}^{T}\Psi_{\theta_{k+1}}x_{k+1}|\mathcal{F}_{k}\Big]=\left(\tilde{A}_{\theta_{k}}x_{k}+B_{\theta_{k}}\nu_{k}\right)^{T}\Pi_{\theta_{k}}\star+\text{tr}\left(\Pi_{\theta_{k}}\right),\end{aligned}
\label{eq:exp_noise}
\end{equation}
where $\Pi_{\theta_{k}}\triangleq\sum\nolimits _{j=1}^{N}\pi_{\theta_{k}j}(x_{k})\Psi_{j}.$
By (\ref{eq:exp_noise}), the objective function in $\mathcal{P}_{2}$ can be expressed
as
\begin{align*}
J_{k}=x_{k}^{T}Q_{\theta_{k}}x_{k}+u_{k}^{T}R_{\theta_{k}}u_{k}+\left(\tilde{A}_{\theta_{k}}x_{k}+B_{\theta_{k}}\nu_{k}\right)^{T}\Pi_{\theta_{k}}\star+\text{tr}\left(\Pi_{\theta_{k}}\right)+\alpha\rho_{k.}
\end{align*}
Thus, the problem $\mathcal{P}_{2}$ is equivalent to
\begin{align}
\boldsymbol{\mathcal{P}_{3}:} & \min_{\nu_{k},\rho_{k}} J_{k}=x_{k}^{T}Q_{\theta_{k}}x_{k}+u_{k}^{T}R_{\theta_{k}}u_{k}+\left(\tilde{A}_{\theta_{k}}x_{k}+B_{\theta_{k}}\nu_{k}\right)^{T}\Pi_{\theta_{k}}\star +\text{tr}\left(\Pi_{\theta_{k}}\right)+\alpha\rho_{k}\nonumber \\
 & \text{s.t.}\; (\ref{sys:djls_noise}), (\ref{eq:p_tran}),\nonumber\\
 &  \quad\;\: G_{j}(\tilde{A}_{\theta_{k}}x_{k}+B_{\theta_{k}}\nu_{k})\le H_{j}-F_{G_{j}w_{k}}^{-1}(\xi)+\mathbf{1}\rho_{k},\, 1\le j\le r,\,\xi\in[0,1],\label{constraint:feasibility}\\
 & \quad\;\: u_{k}=K_{\theta_{k}}x_{k}+\nu_{k},\\
 & \quad\;\:\nu_{k}\in\mathbb{U},\\
 & \quad\;\:\rho_{k}\ge 0 \label{eq:eps_1}.
\end{align}
We denote $J_{k}^{*}$ as an optimal solution of $\mathcal{P}_{3}$. To solve $\mathcal{P}_{3}$, which is a Quadratic Programming (QP) problem, we use
the default QP solver in MATLAB. When we solve $\mathcal{P}_{3}$, if $\rho_{k}$ is zero, for all $k\ge 0$, then
we have solved original $\mathcal{P}_{1}$, and the solution of $\mathcal{P}_{3}$ would be equivalent to the solution of $\mathcal{P}_{1}$. If $\rho_{k}$ is positive, then we
conclude that $\mathcal{P}_{1}$ was infeasible, and the solution of $\mathcal{P}_{3}$ would approximate the solution of $\mathcal{P}_{1}$. This approximation depends on the magnitude of $\rho_{k}$, which need to be lower for better approximation.

\section{Illustrative Example\label{sec:Examples}}
In this section, we illustrate the proposed approach on a dynamics of macroeconomic system. We consider a macroeconomic system based on Samuelson's multiplier accelerator model
that has been considered by  several authors \cite{blairphd}, \cite{blair1975feedback}, \cite{eco_chen2012h}, \cite{eco_costa1999constrained},
\cite{eco_zhang2010fault} etc., to illustrate the respective results on
discrete-time MJLSs. In these works, basically, the discrete-time MJLS modelling of the economic system
was considered to describe the relation between the government expenditure (control input) and the national
income (system state) with different economic situations (modes) and the switching among the modes was modelled as a homogeneous Markov chain; an interested reader can refer to the above references for the detailed mathematical model. In all these references, the authors considered that the transitions among the  economic situations follow a time homogeneous Markov chain. A more realistic scenario is to consider that the transitions among the economic situations depend on the state of the economic system.

An economic system subject to fluctuations with state-dependent economic situation 
can be given by the SDJLS (\ref{sys:djls_noise}) , where $x_k:=[x_k(1)\, x_k(2)]^T\in \mathbb{R}^2$  with $x_{0}=\left[\begin{array}{cc}
2 & 2\end{array}\right]^{T}$. Here $x_{k}(2)$ denotes the national
income, $x_{k}(1)$ denotes the one unit time delayed version of $x_{k}(2)$, $u_k \in \mathbb{R}$ refers to the government expenditure. Let $w_{k}\sim\mathcal{N}(0,\mathbb{I}_2)$, which represents the fluctuations in the national income. Let $\theta_{k}\in\{1,2,3\}$ represent modes of the economic
system under three situations: $i=1$ (normal), $i=2$ (boom), $i=3$
(slump). The parameters of the mode matrices are: 
\[
A_{1}=\left[\begin{array}{cc}
0 & 1\\
-2.5 & 3.2
\end{array}\right],A_{2}=\left[\begin{array}{cc}
0 & 1\\
-4.3 & 4.5
\end{array}\right],A_{3}=\left[\begin{array}{cc}
0 & 1\\
5.3 & -5.2
\end{array}\right],
\]
$B_{\theta_{k}}=\left[\begin{array}{cc}
0 & 1\end{array}\right]^{T}$. The state-dependent transitions of the mode $\theta_{k}$ are given by (\ref{eq:p_tran}),
where

\[
\lambda=\left[\begin{array}{ccc}
0.6 & 0.3 & 0.1\\
0.25 & 0.55 & 0.2\\
0.35 & 0.15 & 0.5
\end{array}\right],\mu=\left[\begin{array}{ccc}
0.67 & 0.17 & 0.16\\
0.30 & 0.47 & 0.23\\
0.26 & 0.10 & 0.64
\end{array}\right] ,
\]
\[
G=\left[\begin{array}{cc}
0 & -1\\
0 & 1
\end{array}\right],\, H\equiv H_{k}=\left[\begin{array}{c}
-2-0.5k\\
5+0.5k
\end{array}\right].
\]
Observe that when the economy (the state $x_k$) satisfies $Gx_k\le H$, the transitions among the modes follow $\lambda$, where the transition probabilities to boom are higher, and also the transition probabilities to slump are lower compared to $\mu$. Thus, we observe the transitions among the economic situations depend on the state of the economy.
Consider the state-feedback control input (\ref{eq:stabilization_input}) with $\Vert\nu_{k}\Vert\le 100$. For pre-stabilization, we consider proposition  \ref{theorem:stabilization_SS} that can be satisfied with

\[
X_{1}=\left[\begin{array}{cc}
 1.3146 & 0\\
0 & 0.7534
\end{array}\right],X_{2}=\left[\begin{array}{cc}
1.9255 & 0\\
0 &  0.7628
\end{array}\right],X_{3}=\left[\begin{array}{cc}
1.1393 & 0\\
0 &  0.1044
\end{array}\right],
\]

\[
Y_{1}=\left[\begin{array}{cc}
  3.2866  &  -2.4108  \\
\end{array}\right],Y_{2}=\left[\begin{array}{cc}
8.2797 & -3.4325\\
\end{array}\right],Y_{3}=\left[\begin{array}{cc}
 -6.0384 & 0.5429 \\
\end{array}\right].
\]
Thus the system is pre-stabilized in the sense of (\ref{eq:msqs}), and the state-feedback gains are given by
\[
K_{1}=\left[\begin{array}{cc}
2.5 & -3.2  \\
\end{array}\right],K_{2}=\left[\begin{array}{cc}
4.3 & -4.5\\
\end{array}\right],K_{3}=\left[\begin{array}{cc}
-5.3 & 5.2 \\
\end{array}\right].
\]

We assume the probabilistic state constraints (\ref{constraint:univariate_final})
with $\xi=0.85$ as a pre-specified monitory target policy at each
$k$. It means that the income (the state) is required to meet the target (the range between the red lines given in figure \ref{fig:ex_stateMC} and figure \ref{fig:eco_income_compare} with a probability $\xi=0.85$, which denotes the level of constraint satisfaction. 

We consider the following parameters for the one-step RHC problem $\mathcal{P}_{3}$: $\alpha=1000,$ $Q_{1}=\mathbb{I}_{2},$ $Q_{2}=1.1\mathbb{I}_{2},$
$Q_{3}=1.2\mathbb{I}_{2},$ $R_{1}=1,$ $R_{2}=1.2,$ $R_{3}=1.3,$
$\Psi_{1}=\mathbb{I}_{2},$ $\Psi_{2}=2\mathbb{I}_{2},$ and $\Psi_{3}=3\mathbb{I}_{2}.$
Using the parameters, we solve $\mathcal{P}_{3}$, at each time $k$. With these parameters we obtain $\rho_k$ as zero for all $k$, which imply that the original problem  $\mathcal{P}_{1}$ was feasible and the solution of $\mathcal{P}_{3}$ is equivalent to the solution of $\mathcal{P}_{1}$. We consider the planning for 5 years of duration, where each unit time represents the period of three months. A sample mode ($\theta_{k}$)  evolution is given in figure \ref{fig:ex_mode}, which shows a sample evolution of an economic situation (normal, bloom or slump).   A corresponding optimal cost $J_{k}^{*}$ is shown in figure \ref{fig:ex_input}. Suppose, if the values of $\rho_k$ are not zero and larger ( $\rho_k >>0$) to make the constraints (\ref{constraint:slack_final}) feasible, then this would make the values of $J_{k}^{*}>>\alpha \rho_k=1000  \rho_k$, because all the remaining terms of the objective function in $\mathcal{P}_3$ are positive.

\begin{figure}[h]
\centering
\begin{minipage}{.5\textwidth}
 \centering
   \includegraphics[width=.8\linewidth]{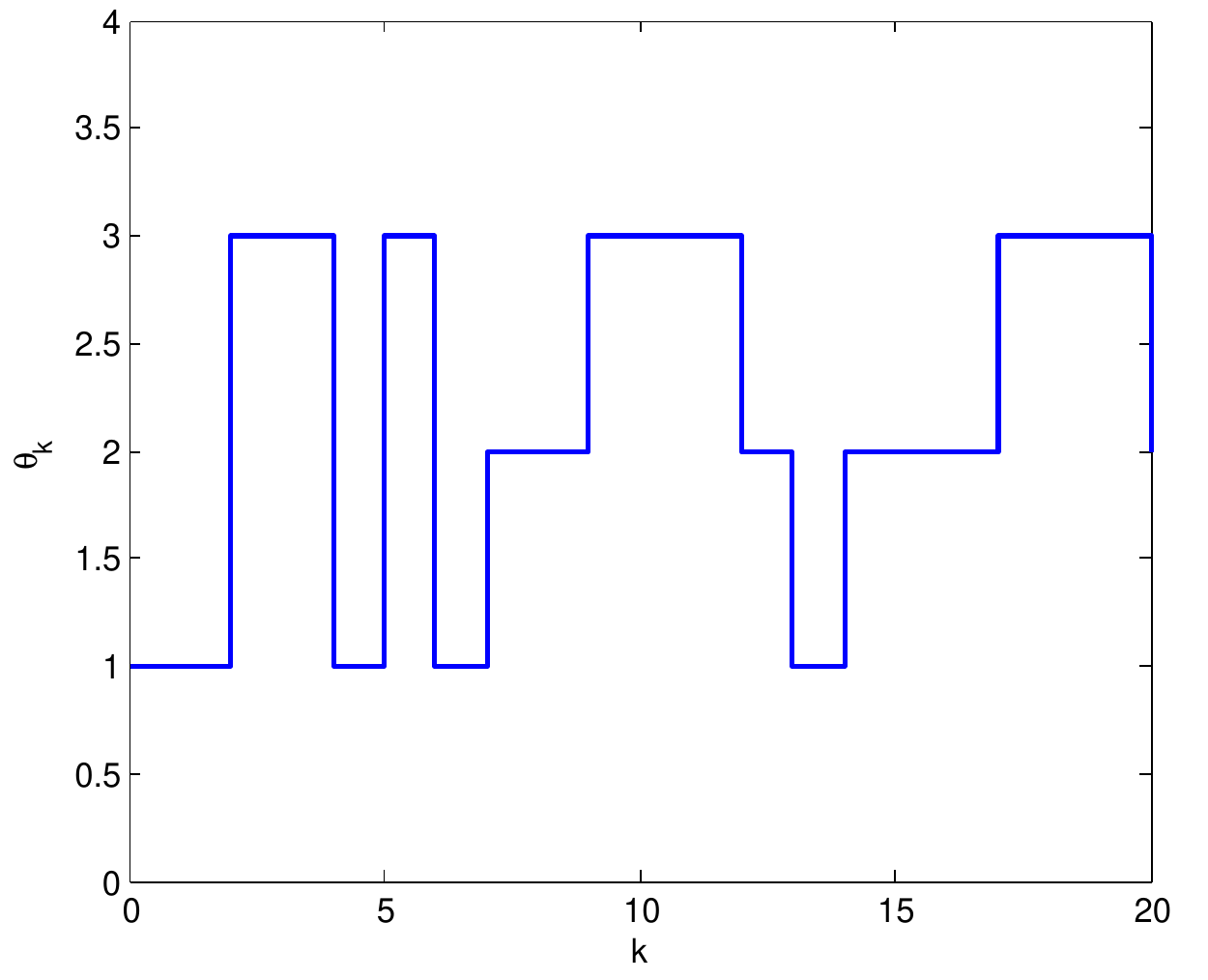}
   \caption{A sample evolution of $\theta_{k}$}
 \label{fig:ex_mode}
\end{minipage}%
\begin{minipage}{.5\textwidth}
    \centering
    \includegraphics[width=.8\linewidth]{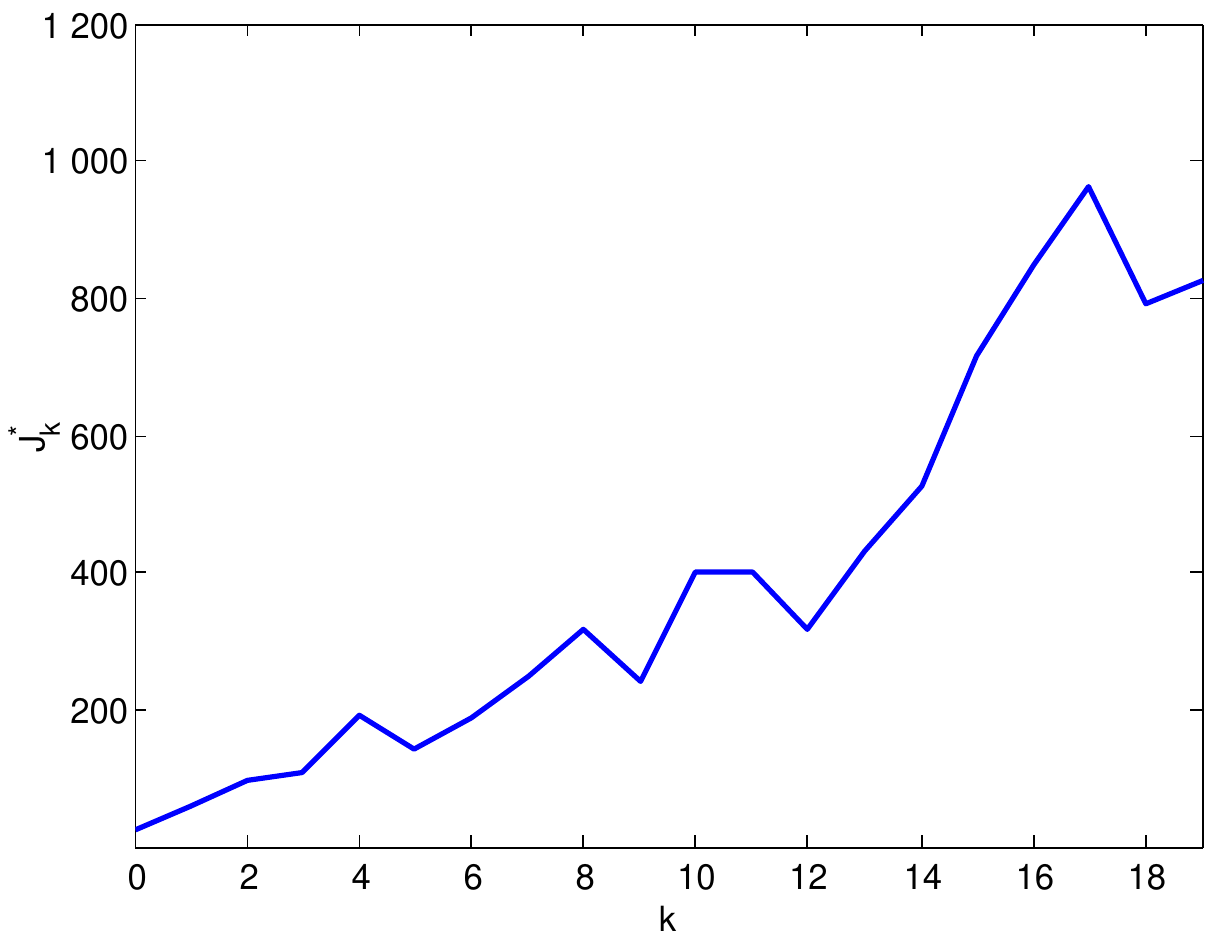}
    \caption{An optimal cost $J_{k}^{*}$}
    \label{fig:ex_input}
\end{minipage}
\end{figure}

To observe the probabilistic state constraints (\ref{constraint:univariate_final}) qualitatively, we performed Monte Carlo simulations for 50 runs and the incomes $(x_k(2))$ are shown in figure  \ref{fig:ex_stateMC}. One can observe occasional constraint violations, because we consider the constraint satisfaction probabilistically. 

To compare the probabilistic constraints  (\ref{constraint:univariate_final}), (\ref{constraint:multivariate_final}), we also solved $\mathcal{P}_3$ with the joint probabilistic constraints (\ref{constraint:multivariate_final}) and the incomes $(x_k(2))$ are shown in figure  \ref{fig:ex_state_joint} that are obtained via Monte Carlo simulations. From figure \ref{fig:ex_stateMC} and \ref{fig:ex_state_joint}, one can observe more constraint violations with individual probabilistic constraints (\ref{constraint:univariate_final}) than (\ref{constraint:multivariate_final}), which is stated in section \ref{sec:prob_constraints}. In figure \ref{fig:ex_state_joint}, we obtain $\rho_k >0$, for some $k$, which is due to smaller feasibility region of (\ref{constraint:chi_square}) compared to (\ref{constraint:det_univariate}) that can be observed by comparing the right hand sides of (\ref{constraint:chi_square}) and (\ref{constraint:det_univariate}).

\begin{figure}[h]
\centering
\begin{subfigure}{.5\textwidth}
  \centering
  \includegraphics[width=.8\linewidth]{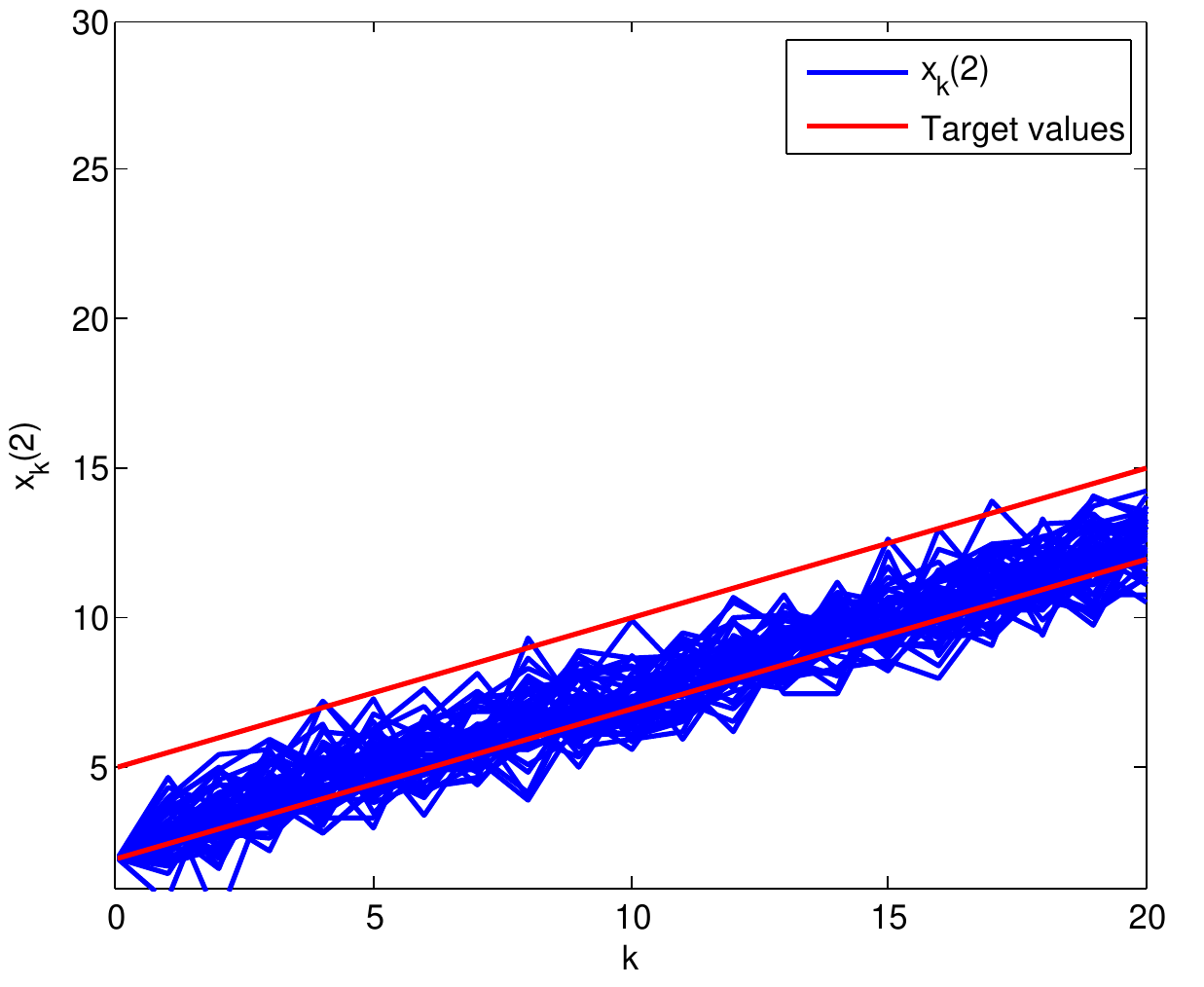}
  \caption{Individual probabilistic constraints (\ref{constraint:univariate_final})}
  \label{fig:ex_stateMC}
\end{subfigure}%
\begin{subfigure}{.5\textwidth}
  \centering
  \includegraphics[width=.8\linewidth]{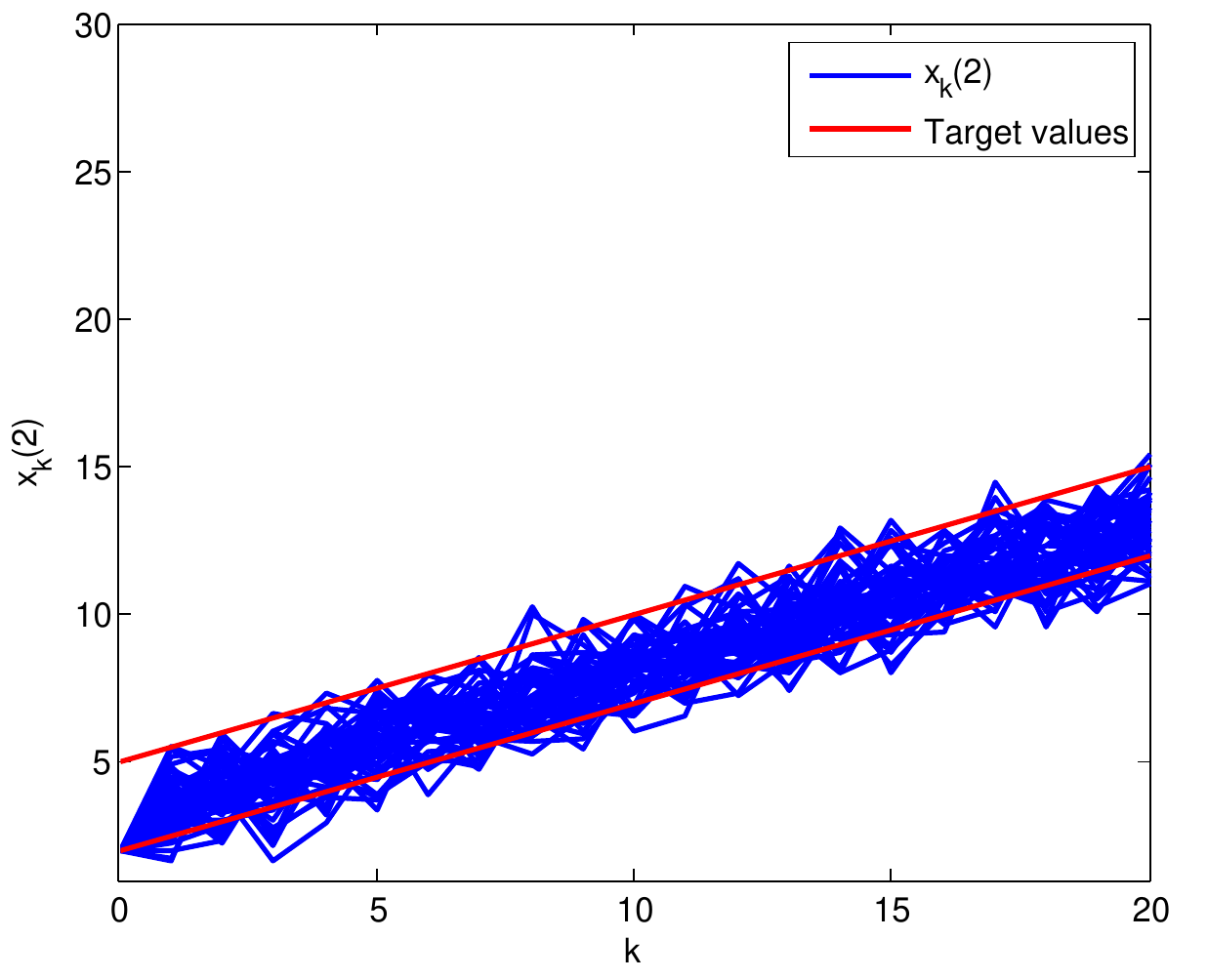}
  \caption{Joint probabilistic constraints (\ref{constraint:multivariate_final})}
  \label{fig:ex_state_joint}
\end{subfigure}
\caption{A qualitative representation of the probabilistic constraints}
\label{fig:ex_st}
\end{figure}


To compare the probabilistic state constraints (\ref{constraint:univariate_final}) with different values of $\xi$ qualitatively, we perform Monte Carlo simulations of 50 runs for values of $\xi$ (0.95, 0.5, 0.3) separately by keeping the remaining parameters same as above, and the incomes are
shown in figure \ref{fig:eco_income_compare}. Observe that the larger the level of constraint satisfaction $\xi$, the better the incomes meet the required target. For this experiment, we obtained $\rho_{k}$ as zero, for all $k$, when $\xi$ is 0.5 and 0.3; $\rho_{k}>0,$ for some $k$, when $\xi$ is 0.95. It implies that the higher the level of constraint satisfaction $\xi$, the more the chance of becoming the original $\mathcal{P}_1$ infeasible. This can also be observed from (\ref{constraint:det_univariate}), where the larger values of $\xi$ makes the right hand side of (\ref{constraint:det_univariate}) smaller,  thus reducing its feasibility region.

\begin{figure}[h]
\centering\includegraphics[width=9cm,height=6cm]{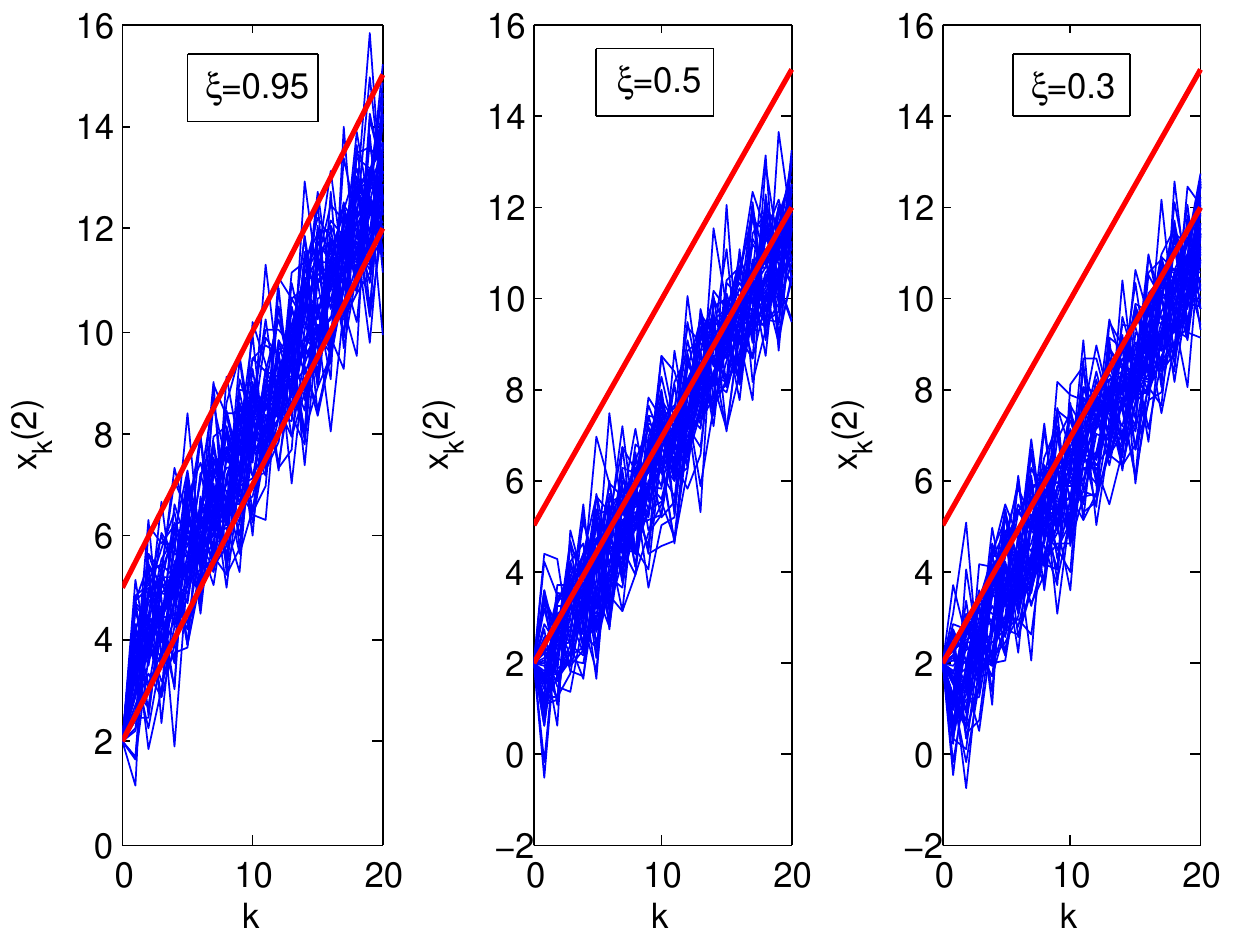}\caption{Illustration of the probabilistic constraints with different values of $\xi$}
\label{fig:eco_income_compare}
\end{figure}

\section{Conclusions}\label{sec:conclusions}
We considered a receding horizon control of discrete-time state-dependent jump linear systems subject to additive stochastic unbounded disturbance with probabilistic state constraints. We used an
affine state-feedback control, and synthesized feedback gains that guarantee
the mean square boundedness of the system  by solving a linear matrix inequality problem off-line.
We obtained sufficient deterministic conditions to satisfy probabilistic
constraints by utilizing inverse cumulative distribution function, and consequently converted
the overall receding horizon problem as a tractable deterministic optimization problem.
Although, it is difficult to guarantee the recursive feasibility in
the case of stochastic unbounded disturbance, we attempted to resolve
this issue with an addition of slack variable to the obtained constraints
with a penalty in the objective function. We performed simulations a macroeconomic system to verify the proposed methodology.

\appendix\label{sec:appendix}
\section{Proof of Proposition 1:}
Consider the given system
\begin{equation}
x_{k+1}=\tilde{A}_{\theta_{k}}x_{k}+B_{\theta_{k}}\nu_{k}+w_{k},\label{sys:proof1}
\end{equation}
where $\theta_{k}$ is defined in (\ref{eq:p_tran}) and $\tilde{A}_{\theta_{k}}=A_{\theta_{k}}+B_{\theta_{k}}K_{\theta_{k}}.$
From remark \ref{remark:input_measurable}, $\{\nu_{k}\}_{k\in \mathbb{N}_{\ge 0}}$
is a stochastic process with the property that for each $k\in \mathbb{N}_{\ge 0}$,
the vector $\nu_{k}$ is ${\mathcal{F}}_{k}$-measurable and $E[\Vert\nu_{k}\Vert^{2}]\leq\delta<\infty$.\\

\noindent One knows that there exist matrices $K_i \in \mathcal{T}^{N}$ and $P_i\in\mathcal{S}_{n}^{N}$, with $P_{i}\succ0$, for $1\leq i\leq N$ such that (\ref{LMI1}) and
(\ref{LMI2}) are verified. Consider $V(x_{k},\theta_{k})=x_{k}^{T}P_{\theta_{k}}x_{k}$. Let $\Upsilon_{(\theta_{k},x_{k})}\triangleq{\textstyle \sum\nolimits _{j=1}^{N}}\pi_{\theta_{k}j}(x_{k})P_{j}$, $\mathcal{L}_{\theta_{k}}\triangleq{\textstyle \sum\nolimits _{j=1}^{N}}\lambda_{\theta_{k}j}P_{j}$,
and $\mathcal{M}_{\theta_{k}}\triangleq{\textstyle \sum\nolimits _{j=1}^{N}}\mu_{\theta_{k}j}P_{j}$. One has

\begin{align*}
\mathbb{E}\big[V&(x_{k+1},\theta_{k+1})|{\mathcal{F}_{k}}\big]-V(x_{k},\theta_{k}) \\ &=x_{k}^{T}\left(\tilde{A}_{\theta_{k}}^{T}\Upsilon_{(\theta_{k},x_{k})}\tilde{A}_{\theta_{k}}-P_{\theta_{k}}\right)x_{k}+2x_{k}^{T}\tilde{A}_{\theta_{k}}^{T}\Upsilon_{(\theta_{k},x_{k})}\mathbb{E}\left[w_{k}|{\mathcal{F}_{k}}\right] +2\nu_{k}^{T}B_{\theta_{k}}^{T}\Upsilon_{(\theta_{k},x_{k})}\tilde{A}_{\theta_{k}}x_{k}\\
 &\quad +2\nu_{k}^{T}B_{\theta_{k}}^{T}\Upsilon_{(\theta_{k},x_{k})}\mathbb{E}\left[w_{k}|{\mathcal{F}_{k}}\right]+\nu_{k}^{T}B_{\theta_{k}}^{T}\Upsilon_{(\theta_{k},x_{k})}B_{\theta_{k}}\nu_{k}+\mathbb{E}\left[w_{k}^{T}\Upsilon_{(\theta_{k},x_{k})}w_{k}|{\mathcal{F}_{k}}\right],\\
  & \leq-\mu\Vert x_{k}\Vert^{2}+2x_{k}^{T}\tilde{A}_{\theta_{k}}^{T}\Upsilon_{(\theta_{k},x_{k})}\mathbb{E}\left[w_{k}|{\mathcal{F}_{k}}\right]
  +2\nu_{k}^{T}B_{\theta_{k}}^{T}\Upsilon_{(\theta_{k},x_{k})}\tilde{A}_{\theta_{k}}x_{k}+2\nu_{k}^{T}B_{\theta_{k}}^{T}\Upsilon_{(\theta_{k},x_{k})}\mathbb{E}\left[w_{k}|{\mathcal{F}_{k}}\right]\\
  &\quad +\alpha_{1}\Vert\nu_{k}\Vert^{2}+\alpha_{2}\mathbb{E}\left[w_{k}^{T}w_{k}|{\mathcal{F}_{k}}\right],
\end{align*}
where $\mu=\underset{1\leq j\leq2}{\min}\left(\underset{1\leq i\leq N}{\min}\left(\lambda_{\min}\left(-\mathcal{E}_{i}^{j}\right)\right)\right),$

\noindent $\alpha_{1}=\underset{}{\max}\left(\underset{1\leq i\leq N}{\max}\left(\lambda_{\max}\left(B_{i}^{T}\mathcal{L}_{i}B_{i}\right)\right),\,\underset{1\leq i\leq N}{\max}\left(\lambda_{\max}\left(B_{i}^{T}\mathcal{M}_{i}B_{i}\right)\right)\right),$ 

\noindent and $\alpha_{2}=\underset{}{\max}\left(\underset{1\leq i\leq N}{\max}\left(\lambda_{\max}\left(\mathcal{L}_{i}\right)\right),\,\underset{1\leq i\leq N}{\max}\left(\lambda_{\max}\left(\mathcal{M}_{i}\right)\right)\right)$.
Because the random vector $w_{k}$ is independent of ${\mathcal{F}}_{k}$
we obtain 
\begin{align*}
\mathbb{E}\big[V&(x_{k+1},\theta_{k+1})|{\mathcal{F}_{k}}\big]-V(x_{k},\theta_{k})\\
 & \leq-\mu\Vert x_{k}\Vert^{2}+2x_{k}^{T}\tilde{A}_{\theta_{k}}^{T}\Upsilon_{(\theta_{k},x_{k})}\mathbb{E}\left[w_{k}\right]+2\nu_{k}^{T}B_{\theta_{k}}^{T}\Upsilon_{(\theta_{k},x_{k})}\tilde{A}_{\theta_{k}}x_{k}+2\nu_{k}^{T}B_{\theta_{k}}^{T}\Upsilon_{(\theta_{k},x_{k})}\mathbb{E}\left[w_{k}\right]\\
 & \quad +\alpha_{1}\Vert\nu_{k}\Vert^{2}+\alpha_{2}\mathbb{E}\left[w_{k}^{T}w_{k}\right],\\
 & \leq-\mu\Vert x_k\Vert^{2}+2\Vert\nu_{k}\Vert\Vert B_{\theta_{k}}\Vert\Vert\Upsilon_{(\theta_{k},x_{k})}\Vert\Vert\tilde{A}_{\theta_{k}}\Vert\Vert x_{k}\Vert+\alpha_{1}\Vert\nu_{k}\Vert^{2}+n\alpha_{2}.
\end{align*}
From the Young's inequality, $2ab\le\kappa a^{2}+\kappa^{-1}b^{2},$
$\forall (a,b)\in\mathbb{R}^{2},$  $\forall \kappa>0$; note that
$\forall\kappa>0$ one has
\begin{align*}
2\Vert\nu_{k}\Vert\Vert B_{\theta_{k}}\Vert\Vert\Upsilon_{(\theta_{k},x_{k})}\Vert\Vert\tilde{A}_{\theta_{k}}\Vert\Vert x_{k}\Vert\leq\kappa\Vert x_{k}\Vert^{2}+\kappa^{-1}\beta^{2}\Vert\nu_{k}\Vert^{2},
\end{align*}
where $\beta=\left(\underset{1\leq i\leq N}{\max}\Vert B_{i}\Vert\right)\cdot\alpha_{2}\cdot\left(\underset{1\leq i\leq N}{\max}\Vert\tilde{A}_{i}\Vert\right)$.
This yields
\begin{align*}
\mathbb{E}\left[V(x_{k+1},\theta_{k+1})|{\mathcal{F}_{k}}\right]-V(x_{k},\theta_{k}) & \leq-(\mu-\kappa)\Vert x_{k}\Vert^{2}+(\alpha_{1}+\kappa^{-1}\beta^{2})\Vert\nu_{k}\Vert^{2}+n\alpha_{2}.
\end{align*}
Take $\kappa\in(0,\mu)$. We have
\begin{align}
\mathbb{E}\left[V(x_{k+1},\theta_{k+1})|{\mathcal{F}_{k}}\right] & \leq(1-c_{2}^{-1}(\mu-\kappa))V(x_{k},\theta_{k})+(\alpha_{1}+\kappa^{-1}\beta^{2})\Vert\nu_{k}\Vert^{2}+n\alpha_{2},\label{equa7}
\end{align}
where $c_{2}=\underset{1\leq i\leq N}{\max}\left(\lambda_{\max}\left(P_{i}\right)\right)$,
with the constraint $(1-c_{2}^{-1}(\mu-\kappa))<1$. Let $q=1-c_{2}^{-1}(\mu-\kappa)$.
Taking expectation on both sides of (\ref{equa7}), we get
\begin{align*}
\mathbb{E}\left[V(x_{k+1},\theta_{k+1})\right]\leq & q\mathbb{E}\left[V(x_{k},\theta_{k})\right]+(\alpha_{1}+\kappa^{-1}\beta^{2})\mathbb{E}\left[\Vert\nu_{k}\Vert^{2}\right]+n\alpha_{2}.
\end{align*}
We obtain recursively
\begin{align*}
\mathbb{E}\left[V(x_{k},\theta_{k})\right] & \leq q^{k}\mathbb{E}\left[V(x_{0},\theta_{0})\right]+(\alpha_{1}+\kappa^{-1}\beta^{2})\sum\nolimits _{p=1}^{k-1}q^{k-p-1}\mathbb{E}\left[\Vert\nu_{p}\Vert^{2}\right] +n\alpha_{2}\sum\nolimits _{p=1}^{k-1}q^{k-p-1},\\
 & \leq q^{k}\mathbb{E}\left[V(x_{0},\theta_{0})\right]+\frac{n\alpha_{2}+\delta(\alpha_{1}+\kappa^{-1}\beta^{2})}{1-q}.
\end{align*}
Finally, we obtain
\begin{align*}
\mathbb{E}\left[\|x_{k}\|^{2}\right] & \leq\frac{c_{2}}{c_{1}}q^{k}\|x_{0}\|^{2}+\frac{n\alpha_{2}+\delta(\alpha_{1}+\kappa^{-1}\beta^{2})}{c_{1}(1-q)},
\end{align*}
where $c_{1}=\underset{1\leq i\leq N}{\min}\left(\lambda_{\min}\left(P_{i}\right)\right)$.
Hence the proof is complete.\hfill{}$\square$

\section{Proof of Lemma \ref{theorem:multivariate}:}

Consider the given probabilistic constraint (\ref{constraint:multivariate}),
\begin{equation}
\text{Pr}\{Gx_{k+1}\le H|\mathcal{F}_k\}\ge\xi.\label{constraint:prob_intheproof}
\end{equation}
Let $H_{k}\triangleq H-G(\tilde{A}_{\theta_{k}}x_{k}+B_{\theta_{k}}\nu_{k}),$
then (\ref{constraint:prob_intheproof}) is equivalent to 

\begin{equation}
\text{Pr}\{Gw_{k}\le H_{k}\}\ge\xi.\label{proof_mv: pconstraint}
\end{equation}
Let $\mathcal{Z}=\{\gamma:\, G\gamma\le H_{k}\},$ then we can rewrite
the above condition as $\text{Pr}\{w_{k}\in\mathcal{Z}\}\ge\xi$, which
is equivalent to

\begin{equation}
\frac{1}{\sqrt{2\pi^{n}}}\int_{\mathcal{Z}}\text{e}^{-\frac{1}{2}\Vert v\Vert^{2}}dv\ge\xi.\label{proof_mv:int}
\end{equation}
Since $\mathcal{Z}$ is a polyhedron, it is difficult to obtain the
closed form of the above integral. So, we consider an inscribed ellipsoidal
approximation of $\mathcal{Z}$ \cite{van2001lmi}, \cite{van2002conic}.
It is reasonable since the level curves of multivariate Gaussian distribution
are ellipsoids. Consider an ellipsoid (since the covariance of $w_{k}$
is identity matrix),

\[
\mathcal{E}_{\delta}\triangleq\{\gamma:\,\Vert\gamma\Vert^{2}\le\delta^{2}\}.
\]
If $\mathcal{E}_{\delta}$ is inscribed in the polyhedron $\mathcal{Z},$
i.e;

\begin{equation}
\mathcal{E}_{\delta}\subset\mathcal{Z},\label{proof_mv:inscribed}
\end{equation}
then $\text{Pr}\{w_{k}\in\mathcal{Z}\}\ge\xi$ is implied by 

\[
\frac{1}{\sqrt{2\pi^{n}}}\int_{\mathcal{E}_{\delta}}\text{e}^{-\frac{1}{2}\Vert v\Vert^{2}}dv\ge\xi.
\]
The above multiple integral can be converted to single integral by
following standard techniques of finding multivariate probability
density function as

\[
\frac{1}{\sqrt{2^{n}}\Gamma(n/2)}\int_{0}^{\delta^{2}}\chi^{\frac{n-2}{2}}\text{e}^{-\frac{1}{2}\chi}d\chi\ge\xi,
\]
that can be rewritten as

\begin{equation}
F_{n,Chi}(\delta^{2})\ge\xi,\label{proof_mv:chi}
\end{equation}
where $F_{n,Chi}(.)$ denotes the Chi-square cumulative distribution
function with $n$ degrees of freedom. The above inequality can be
satisfied with $\delta$ value of $\sqrt{F_{n,Chi}^{-1}(\xi)}.$ By
utilizing the maximization of a liner functional over an ellipsoidal
set, $\mathcal{E}_{\delta}\subset\mathcal{Z}$ can be ensured by \cite{ben2001lectures},\cite{van2001lmi},

\begin{align}
\delta\sqrt{G_{j}G_{j}^{T}}\le H_{kj},\,1\le j\le r,\label{proof_mv:suff}\\
\delta\Vert G_{j}\Vert_{2}\le H_{kj},\,1\le j\le r,
\end{align}
for (\ref{proof_mv:inscribed}), where $G_{j}$ and $H_{kj}$ are $j^{\text{th}}$ rows of $G$ and $H_k$ respectively. Thus (\ref{proof_mv:suff}) is equivalent
to

\[
G_{j}(\tilde{A}_{\theta_{k}}x_{k}+B_{\theta_{k}}\nu_{k})\le H_{j}-\delta\Vert G_{j}\Vert_{2},\,1\le j\le r,
\]
thus completes the proof.\hfill{}$\square$


\bibliographystyle{plain}

\end{document}